%% file: main.tex
\documentclass[acmsmall,screen, nonacm]{acmart}
\AtBeginDocument{%
  }

\author{Margarita Capretto}
\orcid{0000-0003-2329-3769}
\affiliation{\institution{IMDEA Software Institute}\country{Spain}}
\affiliation{\institution{Universidad Politécnica de Madrid}\country{Spain}}
\email{margarita.capretto@imdea.org}
\author{Martín Ceresa}
\orcid{0000-0003-4691-5831}
\affiliation{\institution{IMDEA Software Institute}\country{Spain}}
\email{martin.ceresa@imdea.org@imdea.org}
\author{Antonio {Fernández Anta}}
\orcid{0000-0001-6501-2377}
\affiliation{\institution{IMDEA Software Institute}\country{Spain}}
\affiliation{\institution{IMDEA Networks Institute}\country{Spain}}
\email{antonio.fernandez@imdea.org}
\author{Pedro Moreno-Sánchez}
\orcid{0000-0003-2315-7839}
\affiliation{\institution{IMDEA Software Institute}\country{Spain}}
\affiliation{\institution{VISA Research}\country{Spain}}
\affiliation{\institution{MPI-SP}\country{Germany}}
\email{pedro.moreno@imdea.org@imdea.org}
\author{César Sánchez}
\orcid{0000-0003-3927-4773}
\affiliation{\institution{IMDEA Software Institute}\country{Spain}}
\email{cesar.sanchez@imdea.org@imdea.org}

\authorsaddresses{}
\setcopyright{none}

\input{preamble}

\input{macros}

\input{solidity-highlighting.tex}

\begin{document}
\title{ \bf A Decentralized Sequencer and Data Availability Committee
  for Rollups Using Set Consensus}

\input{Abstract}



\maketitle

\input{Introduction}
\input{Preliminaries}
\input{API}
\input{SeqDC}

\input{SeqDecentralized}
\input{EmpiricalEvaluation}
\input{RelatedWork}
\input{Conclusion}

\vfill
\clearpage

\bibliographystyle{splncs04}
\bibliography{bibfile}

\end{document}

%% file: preamble.tex
\usepackage[utf8]{inputenc}
\usepackage[ruled]{algorithm}
\usepackage[noend]{algpseudocode}
\usepackage{amsmath}
\usepackage{amsfonts}
\usepackage[disable]{todonotes}
\usepackage[ligature,reserved]{semantic}
\usepackage{xspace}
\usepackage{csquotes}

\usepackage[switch]{lineno}

\usepackage{thmtools, thm-restate}

\usepackage{scalerel}
\usepackage{tikz}
\usepackage{adjustbox}

\usetikzlibrary{svg.path}

\makeatletter
\newenvironment{megaalgorithm}[1][htb]{%
    \renewcommand{\ALG@name}{Alg.}
   \begin{algorithm}[#1]%
  }{\end{algorithm}}
\makeatother

 \definecolor{orcidlogocol}{HTML}{A6CE39}
\tikzset{
    orcidlogo/.pic={
        \fill[orcidlogocol] svg{M256,128c0,70.7-57.3,128-128,128C57.3,256,0,198.7,0,128C0,57.3,57.3,0,128,0C198.7,0,256,57.3,256,128z};
        \fill[white] svg{M86.3,186.2H70.9V79.1h15.4v48.4V186.2z}
        svg{M108.9,79.1h41.6c39.6,0,57,28.3,57,53.6c0,27.5-21.5,53.6-56.8,53.6h-41.8V79.1z M124.3,172.4h24.5c34.9,0,42.9-26.5,42.9-39.7c0-21.5-13.7-39.7-43.7-39.7h-23.7V172.4z}
        svg{M88.7,56.8c0,5.5-4.5,10.1-10.1,10.1c-5.6,0-10.1-4.6-10.1-10.1c0-5.6,4.5-10.1,10.1-10.1C84.2,46.7,88.7,51.3,88.7,56.8z};
    }
}
\newcommand\orcidicon[1]{\href{https://orcid.org/#1}{\mbox{\scalerel*{
                \begin{tikzpicture}[yscale=-1,transform shape]
                \pic{orcidlogo};
                \end{tikzpicture}
            }{|}}}}

\newcommand\imarga[1]{\todo[inline]{Marga: {#1}}}

\newtheorem{property}{Property}

\newtheorem{lemma}{Lemma}
\newtheorem{proposition}{Proposition}
 \newtheorem{theorem}{Theorem}

\usepackage{xspace}
\usepackage{url}
\usepackage{paralist}

\usepackage{environ}
\usepackage{wrapfig}

\usepackage{bibunits}
\usepackage{hyperref}

%% file: macros.tex
\newcommand{\arranger}{Arranger\xspace}
\newcommand{\hash}{\mathcal{H}}

\newcommand{\repeattheorem}[1]{%
  \begingroup
  \renewcommand{\thetheorem}{\ref{#1}}%
  \expandafter\expandafter\expandafter\theorem
  \csname reptheorem@#1\endcsname
  \endtheorem
  \endgroup
}

\NewEnviron{reptheorem}[1]{%
  \global\expandafter\xdef\csname reptheorem@#1\endcsname{%
    \unexpanded\expandafter{\BODY}%
  }%
  \expandafter\theorem\BODY\unskip\label{#1}\endtheorem
}

\NewEnviron{replemma}[1]{%
  \global\expandafter\xdef\csname replemma@#1\endcsname{%
    \unexpanded\expandafter{\BODY}%
  }%
  \expandafter\lemma\BODY\unskip\label{#1}\endlemma
}
\newcommand{\repeatatlemma}[1]{%
  \begingroup
  \renewcommand{\thelemma}{\ref{#1}}%
  \expandafter\expandafter\expandafter\lemma
  \csname replemma@#1\endcsname
  \endlemma
  \endgroup
}

\NewEnviron{repproposition}[1]{%
  \global\expandafter\xdef\csname repproposition@#1\endcsname{%
    \unexpanded\expandafter{\BODY}%
  }%
  \expandafter\proposition\BODY\unskip\label{#1}\endproposition
}
\newcommand{\repeatatproposition}[1]{%
  \begingroup
  \renewcommand{\theproposition}{\arabic{proposition}}%
  \expandafter\expandafter\expandafter\proposition
  \csname repprop@#1\endcsname
  \endproposition
  \endgroup
}

\newcommand{\PROP}[1]{\textbf{\textit{#1}}\xspace}
\newcommand{\PrAvailability}{\PROP{Availability}}
\newcommand{\PrTermination}{\PROP{Termination}}
\newcommand{\PrValidity}{\PROP{Validity}}

\newcommand{\PROPsub}[2]{\ensuremath{\bf \textit{\textbf{#1}}_{#2}}\xspace}
\newcommand{\PrIntegrityOne}{\PROPsub{Integrity}{1}}
\newcommand{\PrIntegrityTwo}{\PROPsub{Integrity}{2}}
\newcommand{\PrUniqueBatch}{\PROP{Unique Batch}}

\newcommand{\PrLegality}{\PROP{Legality}}

\newcommand{\PrSBCTermination}{\PROP{SBC-Termination}}
\newcommand{\PrSBCAgreement}{\PROP{SBC-Agreement}}
\newcommand{\PrSBCValidity}{\PROP{SBC-Validity}}

\newcommand{\PrSBCCensorshipResistance}{\PROP{SBC-CensorshipResistance}}
\newcommand{\PrSBCIntegrity}{\PROP{SBC-Integrity}}

\newcommand{\Rtt}{\texttt{R}\xspace}

\newcommand{\LdosTPS}{$200$\xspace}

\reservestyle{\fpm}{\text}
\fpm{fpm[FP mechanism],fpms[FP mechanisms]}
\fpm{fp[FP],fps[FPs]}

\algblockdefx[Upon]{Upon}{EndUpon}%
[1]{{\bf upon} (#1) {\bf do}}%
{{\bf end upon}}

\algblockdefx[When]{When}{EndWhen}%
[1]{{\bf when} (#1) {\bf do}}%
{{\bf end when}}

\algblockdefx[Receive]{Receive}{EndReceive}%
[1]{{\bf receive} (#1) {\bf do}}%
{{\bf end receive}}

\makeatletter
\ifthenelse{\equal{\ALG@noend}{t}}%
  {\algtext*{EndUpon}}
  {}%
\makeatother

\makeatletter
\ifthenelse{\equal{\ALG@noend}{t}}%
  {\algtext*{EndWhen}}
  {}%
\makeatother

\makeatletter
\ifthenelse{\equal{\ALG@noend}{t}}%
  {\algtext*{EndReceive}}
  {}%
\makeatother

\mathlig{<-}{\leftarrow}
\mathlig{==}{\equiv}
\mathlig{++}{\mathop{+\!\!+}}
\mathlig{>>}{\rangle}
\mathlig{~}{\sim}
\reservestyle{\variables}{\text}
\variables{epoch,current,history,theset[the\_set],proposal,valid}

\reservestyle{\setops}{\textsc}
\setops{add,get,Init,Add,Get,BAdd,EpochInc,Broadcast,Deliver,GetEpoch,Propose,Inform,SetDeliver,
  SuperSetDeliver, DoStuff[Start],Consensus}
\setops{translate, Translate}

\reservestyle{\structs}{\text}
\structs{DPO,BAB,BRB,SBC,DSO}
\structs{ArrangerApi[Arranger], SequencerApi[Sequencer], DCsApi[DC
  members], DCApi[DC member]}

\reservestyle{\stmt}{\textbf}
\stmt{call,ack,drop,return,assert,wait, post, inform, trigger, logger}

\reservestyle{\mathfunc}{\mathbf\mathsf}
\mathfunc{send,receive,sendbrb[send\_brb],sendsbc[send\_sbc],
  BLSAggregate[CombineSignatures], sign, broadcast, deliver}

\reservestyle{\messages}{\texttt}
\messages{madd[add],mepochinc[epinc]}

\reservestyle{\api}{\texttt}
\api{apiBAdd,apiAdd,apiGet,apiEpochInc,apiTheSet[{the\_set}],apiHistory}

\reservestyle{\servers}{\texttt}
\servers{S,C,B,M,DAC}

\reservestyle{\schain}{\texttt}
\schain{history,epoch,theset[{the\_set}],add,get,epochinc[{epoch\_inc}],getepoch[{get\_epoch}],tobroadcast[{to\_broadcast}], knowledge, pending, sent, received,proposal,Valid}
\schain{add, translate, allTxs[allTxRqs], pendingTxs[pendingTxRqs], maxTime, maxSize,
  signReq, signResp, hashes, setchain, localsetchain[Setchain],
  invalidId, invalidHash, notContacted[left], signatures} 

\reservestyle{\setvars}{\textit}
\setvars{prop,propset,h,tx,id,sigs,hash, batchId[id],tr}

\reservestyle{\event}{\texttt}
\event{myturn[my\_turn], newepoch[new\_epoch],added,
  timetopost[time\_to\_post],tobatch[to\_batch],
  nextBatch[next\_batch\_id]}

%% file: solidity-highlighting.tex
\usepackage{listings, xcolor}

\definecolor{verylightgray}{rgb}{.97,.97,.97}

\lstdefinelanguage{Solidity}{
	keywords=[1]{anonymous, assembly, assert, balance, break, call, callcode, case, catch, class, constant, continue, constructor, contract, debugger, default, delegatecall, delete, do, else, emit, event, experimental, export, external, false, finally, for, function, gas, if, implements, import, in, indexed, instanceof, interface, internal, is, length, library, log0, log1, log2, log3, log4, memory, modifier, new, payable, pragma, private, protected, public, pure, push, require, return, returns, revert, selfdestruct, send, solidity, storage, struct, suicide, super, switch, then, this, throw, transfer, true, try, typeof, using, value, view, while, with, addmod, ecrecover, keccak256, mulmod, ripemd160, sha256, sha3}, 
	keywordstyle=[1]\color{blue}\bfseries,
	keywords=[2]{set,address, bool, byte, bytes, bytes1, bytes2, bytes3, bytes4, bytes5, bytes6, bytes7, bytes8, bytes9, bytes10, bytes11, bytes12, bytes13, bytes14, bytes15, bytes16, bytes17, bytes18, bytes19, bytes20, bytes21, bytes22, bytes23, bytes24, bytes25, bytes26, bytes27, bytes28, bytes29, bytes30, bytes31, bytes32, enum, int, int8, int16, int24, int32, int40, int48, int56, int64, int72, int80, int88, int96, int104, int112, int120, int128, int136, int144, int152, int160, int168, int176, int184, int192, int200, int208, int216, int224, int232, int240, int248, int256, mapping, string, elem, uint, uint8, uint16, uint24, uint32, uint40, uint48, uint56, uint64, uint72, uint80, uint88, uint96, uint104, uint112, uint120, uint128, uint136, uint144, uint152, uint160, uint168, uint176, uint184, uint192, uint200, uint208, uint216, uint224, uint232, uint240, uint248, uint256, var, void, ether, finney, szabo, wei, days, hours, minutes, seconds, weeks, years},	
	keywordstyle=[2]\color{teal}\bfseries,
	keywords=[3]{block, blockhash, coinbase, difficulty, gaslimit, number, timestamp, msg, gas, sender, sig, value, now, tx, gasprice, origin, add, epochinc, get, setminus, emptyset},	
	keywordstyle=[3]\color{violet}\bfseries,
	identifierstyle=\color{black},
	sensitive=false,
	comment=[l]{//},
	morecomment=[s]{/*}{*/},
	commentstyle=\color{gray}\ttfamily,
	stringstyle=\color{red}\ttfamily,
	morestring=[b]',
	morestring=[b]"
}

%% file: Abstract.tex
\begin{abstract}
  %
  %
  Blockchains face a scalability challenge due to the intrinsic
  throughput limitations of consensus protocols and the limitation in
  block sizes due to decentralization.
  An alternative to improve the number of transactions per second is
  to use Layer 2 (L2) rollups.
  L2s perform most computations offchain using blockchains (L1)
  minimally under-the-hood to guarantee correctness.
  A \emph{sequencer} receives offchain L2 transaction requests,
  batches them, and commits compressed or hashed batches to L1.
  Hashing offers much better compression but requires a data
  availability committee (DAC) to translate hashes back into their
  corresponding batches.
  Current L2s consist of a centralized sequencer which receives and
  serializes all transactions and an optional DAC.
  Centralized sequencers can undesirably influence L2s evolution.
  
  We propose in this paper a fully decentralized implementation of
  a service that combines (1) a sequencer that posts hashes to the L1
  blockchain and (2) the data availability committee that reverses the
  hashes.
  We call the resulting service a (decentralized) \emph{arranger}.
  Our decentralized arranger is based on \emph{Set Byzantine
    Consensus} (SBC), a service where participants can propose sets of
  values and consensus is reached on a subset of the union of the
  values proposed.
  We extend SBC for our fully decentralized arranger.

  Our main contributions are (1) a formal definition of arrangers;
  (2) two implementations, one with a centralized sequencer and another
  with a fully decentralized algorithm, with their proof of
  correctness; and (3) empirical evidence that our solution scales
  by implementing all building blocks necessary to implement a correct
  server.
  \end{abstract}


%% file: Introduction.tex
\section{Introduction}
\label{sec:intro}

%
%
\emph{Distributed ledgers} (also known as \emph{blockchains}) were
first proposed by Nakamoto in 2009~\cite{nakamoto06bitcoin} in the
implementation of Bitcoin, as a method to eliminate trustable third
parties in electronic payment systems.
A current major obstacle for a faster widespread adoption of
blockchain technologies in some application areas is their limited
scalability, due to the limited throughput inherent to Byzantine
consensus
algorithms~\cite{Croman2016ScalingDecentralizedBlockchain,Tyagi@BlockchainScalabilitySol},
and the limitation in the block size due to the desire for
decentralization.
For example, Ethereum~\cite{wood2014ethereum}, one of the most popular
blockchains, is limited to less than 4 blocks per minute, each
containing less than two thousand transactions.

%
%
Layer 2 (L2) rollups provide a faster alternative to blockchains
while still offering the same interface in terms of smart contract
programming and user interaction.
%
%
L2 rollups perform as much computation as possible offchain with the
minimal blockchain interaction---in terms of the number and size of
invocations---required to guarantee a correct and trusted operation.
%
%
L2 rollups work in two phases: users inject transaction requests
communicating with a service called \emph{sequencer}, which orders and
packs the requests into batches.
The sequencer compresses batches and injects the result into an
underlying blockchain (L1).
Once batches are posted to L1, the transaction order inside batches is determined.
Then, the effects of executing transactions are computed
\emph{offchain} by agents called State Transition Functions~(STFs),
which publish the resulting state of the L2 blockchain in L1.
STFs are independent parties that compute L2 blocks from batches and
post them to L1.

There are two main categories of L2 rollups:
\begin{itemize}
\item \emph{ZK-Rollups:} the STFs post zero-knowledge proofs that
  encode the correctness of the transaction batch effects, which are
  verified by the L1 contract that receives the L2 block.
\item \emph{Optimistic Rollups:} STFs post L2 blocks which are
  optimistically assumed to be correct, delegating block validation on
  fraud-proof mechanisms.
\end{itemize}
The most prominent Optimistic Rollups based on their market
share~\cite{l2beat} are Arbitrum One~\cite{ArbitrumNitro}, Optimism
mainnet~\cite{optimism}, and Base~\cite{base}.
Popular ZK-Rollups include Starknet~\cite{starknet}, zkSync
Era~\cite{zksyncera}, and Linea~\cite{linea}.
Fig.~\ref{fig:optimistic-rollups} shows the architecture of ZK Rollups
and Optimistic Rollups.
\begin{figure}[h]
  \centering
  \begin{tabular}{c}
    \includegraphics[scale=0.38]{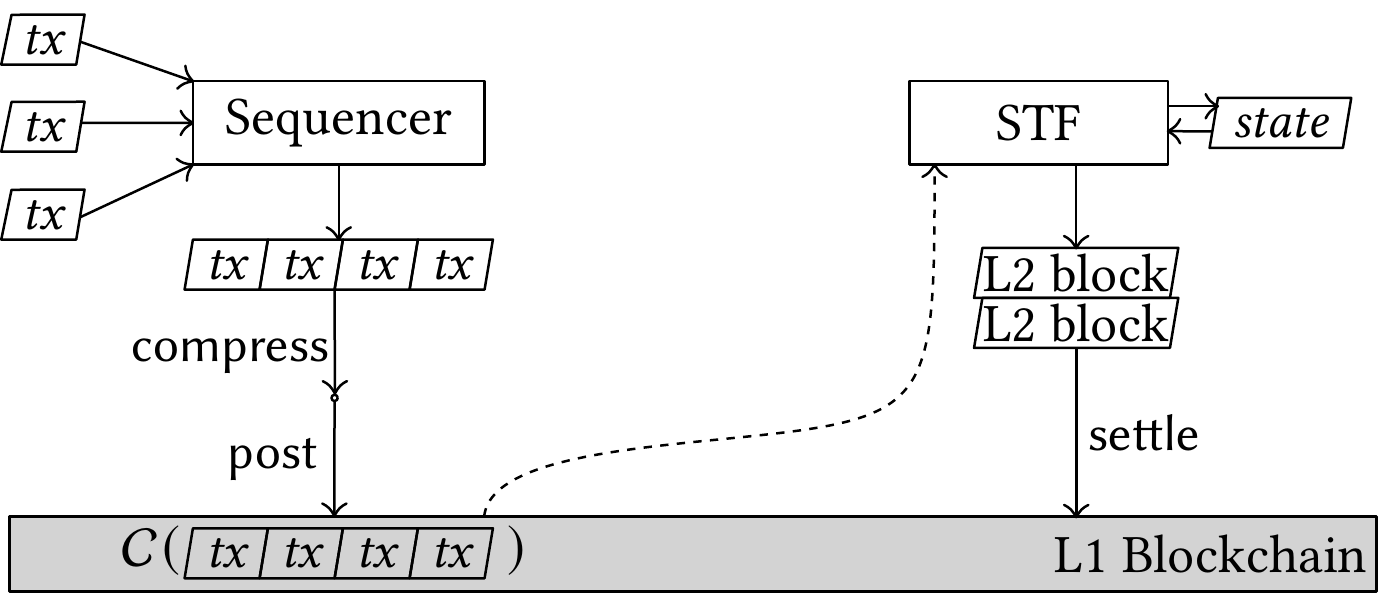}
  \end{tabular}
  \caption{Optimistic and ZK Rollups.}
    \label{fig:optimistic-rollups}
\end{figure}
 
The execution part of L2 rollups can be decentralized\footnote{In the
  sense of \emph{architectural} and \emph{political} decentralization
  as defined in~\cite{buterin17meaning}.} as the role of STF can be open
to many parties.
However, the ordering of transaction requests in most L2 rollups is
decided by a single centralized sequencer node.
These L2 rollups are then non-fully-decentralized, since the sequencer
is a ``\emph{decentralization bottleneck},'' which results in a single
point of trust and failure.
This decentralization bottleneck poses risks such as transaction
censorship or even liveness (fully halting the L2) because the
sequencer can ignore transactions or users, fail to submit batches or
provide incorrect data that does not correspond to batches of valid
transaction requests.
While the limitations of centralized sequencers are known, the
development of decentralized sequencers is still in the early
stages~\cite{motepalli2023sok}.
\emph{We tackle in this paper the problem of decentralized sequencers.}

%
To increase scalability even further, the sequencer in some modern L2
rollups posts hashes of batches---instead of compressed
batches---dramatically reducing the size of the L1 blockchain
interaction.
Using hashes to encode batches requires an additional data
service---called \emph{data availability committee} (DAC)---to
translate hashes into their corresponding batches.
If either the sequencer or the DAC is centralized the solution is
still not a fully decentralized L2 rollup.

ZK-rollups that rely on DACs are known as \emph{Validiums}, which
include Immutable X~\cite{immutablex}, Sophon~\cite{sophon}, and X
Layer~\cite{xlayer}.
Optimistic rollups that use DACs are called \emph{Optimiums}, which
include Mantle~\cite{mantle}, Metis~\cite{metis} and
Fraxtal~\cite{fraxtal}.\footnote{A complete list of Ethereum scaling
  solutions and their current state of decentralization can be found
  in~\cite{l2beat}.}
Fig.~\ref{fig:optimiums} shows the architecture of Optimiums and Validiums.

\begin{figure}[h]
  \centering
  \begin{tabular}{c}
     \includegraphics[scale=0.38]{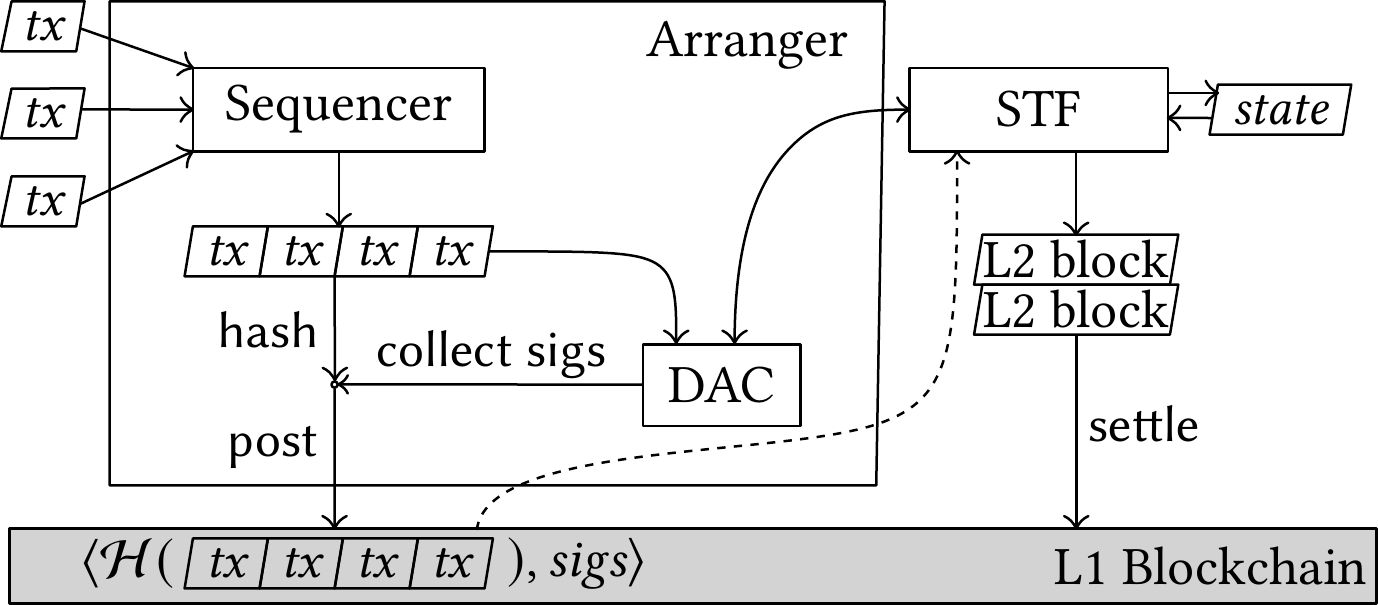}
  \end{tabular}
  \caption{Optimiums and Validiums.}
  \label{fig:optimiums}
\end{figure}

To simplify notation, in the rest of the paper, we use simply L2 to
refer to Optimistic Rollups and ZK Rollups that post batches of
transaction requests as hashes and have a DAC, i.e. Validiums and
Optimiums.
We use the term \emph{arranger} to refer to the combined service
formed by the sequencer and DAC.
Arrangers receive transaction requests, batch them, commit these
batches hashed to L1 (as sequencers do), and can translate the hashes
back into batches (as DACs do).
Arrangers do not execute transactions.

\vspace{0.5em}
\noindent\textbf{Problem.}
%
\emph{To implement a fully decentralized arranger}.

\todo[inline]{Pedro: Before the text mentions that there are some decentralized sequencers but in early stage. Why having a decentralized sequencer is a technical problem? What are the technical issues? (e.g., are we missing security definitions? off-the-shelf approaches like implementing the sequencer's functionality as a multi-party computation would be inefficient. ) }
\vspace{0.5em}
\noindent\textbf{Solution overview.}
%
%
We adopt an open permissioned model~\cite{Crain2021RedBelly} where
permissionless L2 users can issue transaction requests to the
permissioned arranger servers (which we call \emph{replicas}).
%
This model can also be adapted to a permissionless setting with
committee sortition~\cite{gilad2017algorand} without significant
modifications.
We consider a Byzantine failure model~\cite{Lamport1982Byzantine}, in
which there are two types of replicas: \emph{Byzantine}, that behave
arbitrarily, and \emph{honest} replicas, which follow the protocol.
We rely on algorithms whose trust is in the combined power of a
collection of distributed replicas where correctness requires more
than two-thirds of the replicas to run a correct version of the
protocol.

Our main building block is Set Byzantine Consensus
(SBC)~\cite{Crain2021RedBelly}, a Byzantine-resilient protocol where
replicas agree on a set of elements, instead of single elements.
Using SBC, users can inject transaction requests through any replica.
The requests remain unordered until set consensus is executed.
In the consensus phase, each replica proposes a set and replicas
eventually agree on a set of transaction requests, which is guaranteed
to be a sub-set of the union of the sets proposed by all replicas.
We use the consented SBC sets as batches of L2 transaction requests.
In our solution, an honest arranger replica includes as a building
block an honest SBC replica.
Therefore, honest arranger replicas agree on each batch and can
locally compute its hash.
Replicas also post signed hashes to L1 with the guarantee that all
honest replicas can provide, for example to an STFs, the content of a
batch given its hash.

Implementations of
SBC~\cite{Crain2021RedBelly,capretto24improving,ranchalPedrosa2024ZLB}
are reported to be two orders of magnitude faster than binary
consensus (in the volume of transactions processed) by (1) exploiting
the use of only one consensus execution to commit a set of elements,
and (2) the independent validation of the elements in the sets (which
only require checking the validity of request and not its effects).
This enables a more efficient distributed protocol to disseminate
transactions requests, and a much faster validation of each SBC set
compared to validating blockchain blocks which requires executing all
its transactions in order.

It is known that Byzantine-tolerant distributed algorithms do not
compose well in terms of correctness and
scalability~\cite{capretto24improving}, partly because these
algorithms implement defensive mechanisms against misbehaving clients
(local or remote).
Consequently, using an off-the-shelf Byzantine solution as a building
block composed with other building blocks as clients results in a poor
performance.
Instead, we implement here an arranger by combining in each replica a
sequencer and a DAC member, so an honest replica includes an honest
SBC replica and a cooperating honest DAC implementation that can trust
each other.
Our decentralized arranger extends SBC to provide cryptographic
certificates and perform reverse translation of hashes.

\paragraph{\textbf{Contributions.}}
In summary, the contributions of this paper are:

\begin{enumerate}[(1)]
\item A rigorous definition of L2 arrangers;
\item Two implementations of arrangers with their proof of
  correctness, where the more advanced is a fully decentralized
  arranger based on Set-Byzantine-Consensus;
\item An empirical evaluation of all building blocks showing that
  replicas can be efficiently implemented.
\end{enumerate}

\paragraph{\textbf{Structure.}} The rest of the paper is organized as follows.
Section~\ref{sec:prelim} presents the preliminary definitions,
assumptions and the model of computation.
Section~\ref{sec:api} gives a rigorous definition of L2 arrangers.
Section~\ref{sec:seqDC} proposes a basic implementation of a
semi-decentralized arranger.
Section~\ref{sec:seq-decentralized} describes a fully decentralized
arranger using SBC and a proof its correctness.
Section~\ref{sec:empirical} shows empirical evidence that the fully
decentralization implementation using SBC can scale to handle the
current demand of L2s.
Section~\ref{sec:related-work} discusses the related work.
Finally, Section~\ref{sec:conclusion} concludes the paper.


%% file: Preliminaries.tex
\section{Definitions and Model of Computation}\label{sec:prelim}

We state now our assumptions about L1s, describe the model of
computation, and present an overview of the Set Byzantine Consensus
protocol.

\subsection{Model of Computation.}

We adopt a standard distributed system
model~\cite{Castro2002Practical,Kotla2010Zyzzyva} where messages are
delivered within a bounded but unknown time limit and clocks are
almost in synchrony, usually referred as partial
synchrony~\cite{Dwork1988Consensus}.

Our system comprises arrangers replicas and clients.
There are two types of clients:
L2 users sending L2 transaction requests to arranger replicas, and
    STFs requesting arranger replicas the translation of
    hashes posted in L1
    into batches.
Arranger replicas can be classified as either \emph{honest}, meaning
they adhere to the arranger protocol, or \emph{Byzantine}, which
behave arbitrarily~\cite{Lamport1982Byzantine}.
We assume that an upper bound \(f\) in the number of Byzantine
replicas is known.

We consider a public-key infrastructure (PKI) that associates replica
and client identities with their public keys, and that is common to
all replicas and clients.
Replicas and clients use a function \(\<sign>\) to sign elements with
their secret key.
L2 users can create \emph{valid} transaction requests and arranger
replicas cannot impersonate clients.
Valid transaction requests are those that have been correctly signed,
so arranger replicas can locally check transaction requests validity
using public-key cryptography.

Arranger replicas use a known collision-resistant hash function
\(\hash{}\) to hash batches.

\subsection{Assumptions about L1}

We assume that the L1 ensures both liveness and safety.
Specifically, while the system tolerates temporary censorship of L1
transaction requests and reordering of L1 transactions, it guarantees
that every transaction submitted to L1 is eventually processed
correctly.
The L1 includes a smart contract \<logger> that arranger replicas
use to post hashed batches of transaction requests.
The \<logger> smart contract knows the public key of all arranger
replicas and the bound \(f\) on the number of Byzantine replicas.
The \<logger> smart contract only accepts new hashes from arranger
replicas that are signed by at least \(f+1\) arranger replicas,
guaranteeing that at least one honest replica signed it.

\subsection{Set Byzantine Consensus}\label{sec:setconsensus}
Set Byzantine Consensus (SBC)~\cite{Crain2021RedBelly} is a variant of
Byzantine consensus where instead of proposing and agreeing on a
single value, SBC replicas propose sets of values and agree on a
non-empty subset of the union of the proposed sets.
This allows committing more elements per consensus instance, improving
the throughput, as empirically demonstrated in practical
implementations~\cite{Crain2021RedBelly,capretto24improving,ranchalPedrosa2024ZLB}.
Intuitively, SBC increases throughput by running many binary consensus
simultaneously to determine the inclusion of elements in proposed sets.
We use SBC to determine the set of transaction requests to be included
in a batch.

SBC replicas provide two end-points.
\begin{itemize}
\item \(\<Add>(e)\) is used to submit an element \(e\).
  Honest replicas include \(e\) in their proposed set.
  Once \(e\) is known by all honest replicas, it is guaranteed to
  eventually be included in a decided set.
\item \(\<SetDeliver>(i,E)\) notifies when the \(i\)-th
  round of SBC finishes with $E$ as the decided set.
\end{itemize}

The properties of set consensus relevant to our implementation of
arrangers in Section~\ref{sec:seq-decentralized} are:
\begin{itemize}
\item \PrSBCTermination: every honest SBC replica eventually decides a
  set of elements in each SBC round.
\item \PrSBCAgreement: in each SBC round, all honest replicas decide
  the same set.
\item \PrSBCValidity: the decided set in a given SBC round is a
  non-empty subset of the union of the proposed sets in that round and
  contains only valid elements.
\item \PrSBCCensorshipResistance: elements known and proposed by all
  honest replicas are eventually included in the set decided in some
  round.
\item \PrSBCIntegrity: no element appears in more than one decided set
  across all SBC rounds.
\end{itemize}

There are several implementations of SBC, including
Redbelly~\cite{Crain2021RedBelly}, Setchain~\cite{capretto24improving}
and ZLB~\cite{ranchalPedrosa2024ZLB}.
These implementations assume that the Byzantine replicas are less than
one third.
However, there are nuances in how these implementations ensure certain
properties. 
For example, ZLB does not mention \PrSBCCensorshipResistance, and
a minor modification is required for RedBelly and ZLB to ensure
\PrSBCIntegrity: replicas must keep a log of all elements included
in previously decided sets to detect and remove duplicates.


%% file: API.tex
\section{Arranger}
\label{sec:api}

In this section, we define the concept of arranger, the service in
charge of both (1) receiving and serializing transaction requests,
packing them into batches and efficiently posting them as hashes into
L1; and (2) making the data available.
Our arranger model seamlessly fits into the model of existing
Optimiums and Validiums (see Fig.~\ref{fig:optimiums}).

\subsection{Arranger API}

Arranger replicas provide two end-points:
\begin{itemize}
\item \<add>(\<tr>) used by L2 users to submit a transaction
  request \<tr>, and
\item \<translate>($\<id>,\<h>$) used by STFs or any other external
  users to request the batch of transaction requests corresponding to
  a hash \(\<h>\) with identifier \(\<id>\).
  If identifier \(\<id>\) does not match any batch, the arranger
  returns error \<invalidId>.
  If there is a batch \(b\) with identifier \(\<id>\), the arranger
  returns \(b\) if \(b\) hashes to \(\<h>\) or error \<invalidHash>
  otherwise.
\end{itemize}

When the arranger receives enough transaction requests or a timeout is
reached, all honest arranger replicas decide a new batch $b$, order
transaction requests in \(b\), assign an identifier $\<id>$ to \(b\),
compute $h = \hash{(b)}$, and create a \emph{batch tag} $(\<id>,h)$.
Then, all honest replicas sign the new batch tag and propagate their
signatures to the other replicas.
When enough signatures are collected, a combined signature
\(\sigma\)\footnote{A combined signature not only contains the result
  of combining all signatures but also an identifier of each signer.}
is attached to the batch tag to form a \emph{signed batch tag}
$(\<id>,\<h>,\sigma)$ and the signed batch tag is posted to L1.
The \<logger> L1 smart contract accepts signed batch tags validating
only the signature of the batch.

\subsection{Arranger Properties}
\label{sec:api-properties}
We now define the properties that characterize a \emph{correct}
arranger.
These properties include all \emph{basic properties} and some safety
properties of ideal arrangers introduced informally
in~\cite{motepalli2023sok}.\footnote{None of the L2s studied
  in~\cite{motepalli2023sok} implement both a decentralized sequencer
  and a decentralized DAC.}

We introduce some definitions about signed batch tags.
A signed batch tags is called \emph{certified} when is contains at
least \(f+1\) arranger replicas signatures.
Certified batch tags are guaranteed to include at least one honest
replica signature, as the required number of signatures exceeds the
maximum number of Byzantine replicas assumed.

A certified batch tag $(\<id>,\<h>,\sigma)$ is legal if its
corresponding batch \(b\) satisfies the following properties:
\begin{itemize}
\item \PrValidity: Every transaction request in $b$ is a valid
  transaction request added by an L2 user.\footnote{A transaction
    request is valid when it is properly formed and signed by the
    originating L2 user, which can be locally verified by arranger
    replicas.}
  \item \PrIntegrityOne: No transaction request appears twice in $b$.
  \item \PrIntegrityTwo: No transaction request in $b$ appears in a
    legal batch tag previously accepted by the \<logger> smart
    contract.
  \end{itemize}

We require all certified batch tags posted by correct arrangers to be
legal.

\begin{property}[\textup{\PrLegality}]
  \label{prop:legality}
  Every certified batch tag posted by the arranger is a legal batch
  tag.
\end{property}

Arranger replicas may post multiple signed batch tags with the same
identifier because they concurrently try to post the next batch, but
the \<logger> accepts only the first one posted.
A batch tag can be part of two signed batch tags when signed by a
different subset of arranger replicas.
However, certified batch tags with the same identifier must have the
same batch and generate the same hash.
This ensures that the reorder of L1 transactions cannot influence the
evolution of the L2 blockchain.

\begin{property}[\textup{\PrUniqueBatch}]
  \label{prop:uniqueBatch}
  Let $(\<id>, \<h>_1, \sigma_1)$ and $(\<id>, \<h>_2, \sigma_2)$ be
  two certified batch tags with the same identifier $\<id>$.
  Then, \(\<h>_1 = \<h>_2\).
\end{property}

The following properties of correct arrangers prevent censorship and
guarantee data availability.

\begin{property}[\textup{\PrTermination}]
  \label{prop:basic}
  All valid transaction requests added to honest replicas eventually
  appear in a posted legal batch tag accepted by the \<logger> smart
  contract.
\end{property}
\begin{property}[\textup{\PrAvailability}]
  Every posted legal batch tag can be translated into its batch by an
  honest replica.
\end{property}
 
\PrAvailability is expressed formally as follows. Let
 $(\<id>, \<h>, \sigma)$ be a legal batch tag posted by the arranger,
 s.t. $\<h>=\hash(b)$.
 Then, some honest replica will return $b$ when requested
 $\<translate>(\<id>,\<h>)$.
This prevents halting the L2 blockchain by failing to provide
batches of transaction requests from hashes.
If a Byzantine replica returns a batch \(b' \neq b\), then $b'$ cannot
hash to \(\<h>\) because of collision resistance, and the client can
locally detect this violation by computing \(\hash(b')\).

\PrLegality, \PrUniqueBatch, and \PrAvailability are safety properties
and \PrTermination is a liveness property.
Altogether, they characterize correct arrangers.

The combination of \PrTermination, \PrUniqueBatch, and \PrAvailability
guarantee that for correct arrangers all valid transaction requests
added to honest arranger replicas are eventually executed in the L2
blockchain.
When a valid transaction request $\<tr>$ is added to an honest arranger
replica, \PrTermination ensures that $\<tr>$ will eventually be included
in a legal batch tag $(\<id>, \<h>, \sigma)$ posted in L1.
Then, when an STF requests the transaction requests in a legal batch
tag with identifier \<id>, \PrAvailability ensures that an honest
arranger replica responds with the corresponding batch of transactions
$b$.
\PrUniqueBatch guarantees that all legal batch tags with identifier
\<id> correspond to the same batch of transactions.
This means that $b$ contains $\<tr>$.
Since the STF executes the entire batch $b$, transaction $\<tr>$ is
executed.
Therefore, \emph{correct arrangers offer censorship
  resistance}.
Additionally, since STFs only process legal batch tags, \PrLegality
prevents any transaction from being executed twice (properties
\PrIntegrityOne and \PrIntegrityTwo of legal batch tags), and
guarantees that only valid transaction requests are executed (property
\PrValidity of legal batch tags).
However, the properties listed above do not guarantee any order
of transactions posted by L2 users.
%


%% file: SeqDC.tex
\section{Arranger \#1: Centralized Sequencer + Decentralized DAC}
\label{sec:seqDC}

We present first a \emph{semi-decentralized} distributed
implementation of arrangers where a single replica acts as the
sequencer and the remaining replicas implement a decentralized DAC.
The centralized sequencer implements operation \<add> locally, creates
batches of transaction requests added by L2 users, communicates these
batches to all DAC replicas, collect their signatures and post legal
batch tags invoking the \<logger> smart contract in L1.
DAC replicas implement \<translate> to provide the inverse
resolution of hashes posted by the sequencer back into readable
batches of transaction requests.
Fig.~\ref{fig:arranger_seqDC} shows this semi-decentralized arranger
architecture.
\begin{figure}[h!]
    \centering
    \includegraphics[scale=0.40]{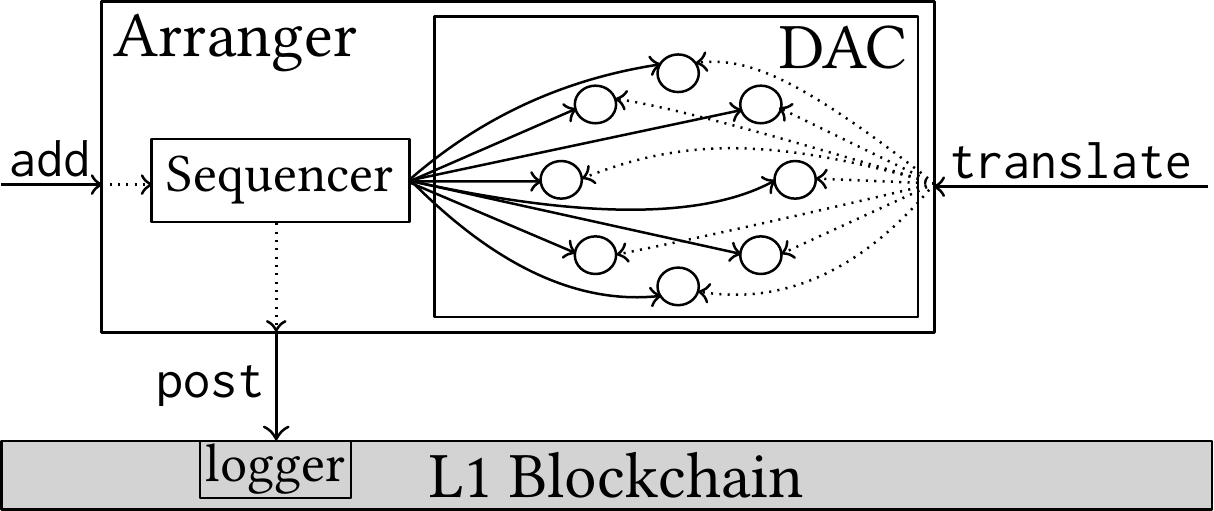}
    \caption{Arranger \#1:  centralized sequencer + decentralized DAC.}
    \label{fig:arranger_seqDC}
\end{figure}

\subsection{Implementation}
\input{algorithms/sequencer}
\input{algorithms/DC-member}
%
Alg.~\ref{alg:sequencer} and Alg.~\ref{alg:dc-member} show the
pseudo-code of the sequencer and DAC member, respectively.
The \<logger> smart contract is invoked only by the sequencer to store batch
tags using function \<post>.
For simplicity, the sequencer uses a multicast service,
\texttt{\textbf{DC}}, to send messages to all DAC members, which can
be implemented with unicast channels simply by the sequencer
contacting each DAC replica in parallel.
DAC replicas answer back to the sequencer using a single unicast
message.


Alg.~\ref{alg:sequencer} presents the sequencer pseudo-code.
This algorithm maintains a set \<allTxs> to store all valid
transaction requests added, a sequence of transaction requests
\<pendingTxs> for transaction requests added but not posted yet, and a
natural number \<batchId> as the identifier of the next batch to be
posted.
%
L2 users add transaction requests to the arranger by invoking function
\<add> in the sequencer.
Function \<add> accumulates new valid transaction requests into sequence
\<pendingTxs> and set \<allTxs>.
When there are enough transaction requests in \<pendingTxs> or enough
time has elapsed since the last posted batch, a new event
\<timetopost> is triggered (line~\ref{seq-when}) and
Alg.~\ref{alg:sequencer} generates a new batch with all transaction
requests currently stored in \<pendingTxs>
(line~\ref{seq-when-save-pendingTxs}).
Then, the sequencer sends the fresh new batch with its identifier and
hash to all DAC replicas (line~\ref{seq-send-h}) and waits to receive
$f+1$ signatures, which ensures that at least one honest DAC replica
signed the batch and this replica can translate the batch upon
request.
Then, Alg.~\ref{alg:sequencer} aggregates all signatures received,
generates a batch tag including the combined signature and posts the
batch tag in the \<logger> smart contract in L1 (line~\ref{seq-post}).
Finally, Alg.~\ref{alg:sequencer} cleans \<pendingTxs> and increments
\<batchId> (lines~\ref{seq-when-clean-pendingTxs}
and~\ref{seq-increment-batchid}).
%

Alg.~\ref{alg:dc-member} presents the pseudo-code of honest DAC
replicas, which maintain a map \<hashes> mapping pairs of identifiers
and hashes $(\<batchId>,h)$ into batches $b$, such that
$\hash(b) = h$.
When the sequencer sends requests to sign a batch $(b,\<batchId>,h)$,
honest DAC replicas check that $h$ corresponds to hashing $b$
(line~\ref{dc-member-sign-assert}).
If the hash matches, the replica adds the new triplet to \<hashes>,
signs it and send the signature back to the sequencer.
Clients, such as STFs, can request the batch corresponding to a given
hash \(h\) in a batch tag with identifier \(\<batchId>\).
These clients perform the request by invoking
$\<translate>(\<batchId>,h)$.
Honest DAC replicas return the value stored in
$\<hashes>(\<batchId>,h)$, or an error when they do not know the hash
or the identifier does not match the hash.

Clients requesting to translate hashes from legal batch tags
can obtain the corresponding batch by contacting at most \(f+1\) DAC
replicas.
Answers can be locally verified by clients by checking that the batch
returned by a DAC replica corresponds to the hash \(h\) provided, as
we are assuming that the hash function is collision-resistant.
Therefore, clients need to contact multiple replicas only when they do
not receive a response or when the response does not match the hash,
but it suffices one correct answer by one replica, which can be
locally authenticated.
Since legal batch tags are signed by at least one honest replica, that
replica can correctly translate hashes (see
Lemma~\ref{l:semiTranslate}).
In the worst case, clients translating batches may first contact all
\(f\) Byzantine DAC replicas, who do not respond with the
corresponding batch, before requesting and receiving the desired batch
from an honest replica.

Clients can employ different strategies for contacting DAC members,
such as reaching out to all replicas in parallel or contacting them
sequentially until the desired batch is obtained.
Studying the trade-offs between these strategies is beyond the scope
of this paper.



\subsection{Proof of Correctness}\label{sec:seqDC-props}

We prove that an arranger consisting of the following components is
correct:
\begin{itemize}
\item a single sequential replica, denoted with \<S>, running
  Alg.~\ref{alg:sequencer}, and
\item a data availability committee, denoted with \<DAC>, composed of
  $n$ replicas.
  Out of these, at most $f < \frac{n}{2}$ replicas may be Byzantine,
  and the remaining \(n-f\) replicas are honest and implement
  Alg.~\ref{alg:dc-member}.
\end{itemize}

We assume that a majority of \<DAC> replicas are honest to guarantee
\PrTermination~(see Lemma~\ref{lem:semiTermination} below), and
discuss the trade-off between this trust assumption and liveness in
Section~\ref{sec:trustVSliveness}.

The arranger satisfies \PrLegality and \PrUniqueBatch
(Lemmas~\ref{l:semiLegal} and~\ref{l:semiUnique} below), because
the sequencer \<S> is the only replica that generates certified batch
tags and all batch tags \<S> generates are legal and have different
identifiers.

\begin{lemma}
  \label{lem:semiCert}
  Certified batch tags are generated only by sequencer \<S>.
\end{lemma}

\begin{proof}
  Certified batch tags require \(f+1\) signatures, so at least one
  honest DAC replica signed the batch.
  Honest DAC replicas only sign batches received from the sequencer
  \<S>.
\end{proof}

\begin{lemma}\label{l:semiLegal}
  All certified batches generated and posted by sequencer \<S> are
  legal.
\end{lemma}
\begin{proof}
  Following Alg.~\ref{alg:sequencer}, \<S> generates batches taking
  transaction requests from $\<pendingTxs>$, and thus, transaction
  requests come from L2 users invocations to \<add>.
  Function \<add> guarantees that each transaction request in
  $\<pendingTxs>$ is valid and new.
  Moreover, $\<pendingTxs>$ is emptied after generating batches so
  transaction requests are in only one batch generated by \<S>.
  Before posting batch tags, \<S> gathers at least \(f + 1\)
  signatures, and thus, \<S> posts only legal batch tags.
\end{proof}

\begin{lemma}\label{l:semiUnique}
  All certified batches generated by sequencer \<S> have different
  identifier.
\end{lemma}
\begin{proof}
  Once sequencer \<S> creates a certified batch tag it increments
  variables \<batchId> (see line~\ref{seq-increment-batchid} in
  Alg.~\ref{alg:sequencer}), which is used to as the identifier for
  the next batch tag.
  Sequencer \<S> never decreases variable \<batchId>, so each batch tag
  has a different identifier.
\end{proof}

The following lemma implies \PrTermination.

\begin{lemma}
  \label{lem:semiTermination}
  Valid transaction requests added through \<add> eventually appear in
  a legal batch tag accepted by the \<logger> smart contract.
\end{lemma}

\begin{proof}
  Let \(\<tr>\) be a valid transaction request.
  The first time \<add>\((\<tr>)\) is invoked, \<S> appends $\<tr>$ to
  $\<pendingTxs>$ where \(\<tr>\) remains until it is added to a
  batch.
  At some point after $\<tr>$ was added, event $\<timetopost>$ is
  triggered, so \<S> generates a new batch with all transaction
  requests currently in $\<pendingTxs>$.
  Then, \<S> sends to all \<DAC> members the batch
  (lines~\ref{seq-when}-\ref{seq-send-h}) and waits for $f+1$
  responses (line~\ref{seq-wait}).
  Since $f<n/2$, there are at least $f+1$ honest \<DAC> replicas
  executing Alg.~\ref{alg:dc-member}, all of which sign the hash and
  send it back to \<S>
  (lines~\ref{dc-member-receive}-\ref{dc-member-receive-end}).
  Therefore, \<S> eventually receives enough signatures and posts the
  legal batch tag containing $\<tr>$ (line~\ref{seq-post}).
  Since sequencer \<S> only generates one certified batch tag per
  identifier (see Lemma~\ref{l:semiUnique}), the \<logger>
  contract accepts the legal batch tag containing \(\<tr>\).
\end{proof}

Finally, we show each certified batch tag can be resolved by at least
one honest \<DAC> member.
Let \(m\) be a \<DAC> member that signed $(\<batchId>,h)$, the
result of executing $m.$\<translate>$(\<batchId>, h)$ is a batch $b$
such that $\hash (b) = h$.
Since all certified batch tags are signed by at least one honest DAC
member, it follows that arrangers consisting of \<S> and \<DAC>
satisfy \PrAvailability.

\begin{lemma}\label{l:semiTranslate}
  Honest \<DAC> members resolve certified batch tags.
\end{lemma}

\begin{proof}
  Let $m$ be an honest \<DAC> replica.
  Following Alg.~\ref{alg:dc-member}, \(m\) only signs hashes
  answering requests from \<S>.
  Upon receiving a request from \<S> to sign $(h,b,\<batchId>)$, after
  checking that $h$ correspond to \(b\) and \(\<batchId>\) has not been
  used, \(m\) saves the data received, signs the request and answers
  to \<S>.
  As no entry in $\<hashes>$ is ever deleted, whenever
  $m.$\<translate>$(\<batchId>,h)$ is invoked after \(m\) has singed
  the request, $m$ uses map $\<hashes>$ to find batch $b$ and returns
  it.
\end{proof}

Alg.~\ref{alg:sequencer} plus Alg.~\ref{alg:dc-member} guarantee that
\<S> and \<DAC> form a correct arranger, which follows from the previous
lemmas.
\begin{theorem}
  Arrangers composed by sequencer \<S> and DAC \<DAC> are correct
  arrangers.
\end{theorem}

Finally, note that Alg.~\ref{alg:sequencer} orders transaction
requests according to the invocation order to \<add>, so transaction
requests are ordered according to a first-come, first-served policy.

\subsection{Trust Assumptions vs Liveness}
\label{sec:trustVSliveness}

\imarga{seems like requiring an honest majority is not a good
  idea. L2beat people claim that DACs should require \(1/3\) or less
  honest members. l2beat . com/scaling/projects/immutablex \# da-layer
}
%
Arrangers of some L2s also use a centralized sequencer ordering
transaction requests and a DAC translating hashes back to batches.
The main difference between the implementation described in this
section and the implementation currently used in most Validiums and
Optimiums is that the latter requires only a minority of DAC members
to be honest, whereas our implementation requires a majority.
Having a minority of honest DAC replicas and an honest sequencer is
enough to guarantee all safety properties of a correct arranger but
not to ensure progress.

Arrangers with just a minority of honest DAC members violate property
\PrTermination because collecting signatures from honest DAC replicas
is not sufficient to generate certified batch tags.
As a result, some of those L2s (like Arbitrum AnyTrust) implement a
fallback mechanism: the sequencer can post a compressed version of
batches if it is unable to collect enough DAC signatures.

In contrast, this is not necessary in our solution, as we assume that
a majority of DAC replicas are honest guaranteeing \PrTermination~(see
Lemma~\ref{lem:semiTermination} above).
%
This highlights a trade off between trust assumptions and liveness
properties.

Regardless of the trust assumption for DAC replicas, both
semi-decentralized arrangers rely on the centralized sequencer as a
single point of trust and failure.
In the next section, we eliminate this single point of trust and
failure by presenting a fully decentralized arranger using SBC as the
main building block.


%% file: algorithms/sequencer.tex
\begin{figure}[h]
  \begin{algorithm}[H]
    \caption{\small Sequencer.}
    \label{alg:sequencer}
    \small
    \begin{algorithmic}[1]
  	\State \textbf{init:} $\<allTxs> \leftarrow \emptyset $,
        $\<pendingTxs> \leftarrow \epsilon$, $\<batchId> \leftarrow 0$~\label{seq-init}
        \Function{\textnormal{\<add>}}{\textit{tr}} \label{seq-add}
          \State \<assert> $\<valid>(\<tr>) $ and $\<tr> \notin \<allTxs>$ \label{seq-add-assert}
          \State $\<pendingTxs> \leftarrow \langle \<pendingTxs>, \<tr> \rangle$ \label{seq-add-pendingTxs}
          \State $\<allTxs> \leftarrow \<allTxs> \cup \{\<tr>\}$ \label{seq-add-allTxs}
          \State \Return
        \EndFunction \label{seq-add-end}
        \Upon{\<timetopost>}\label{seq-when} 
          \State $b \leftarrow \<pendingTxs>$ \label{seq-when-save-pendingTxs}
          \State $h \leftarrow \hash(b)$ \label{seq-when-hash-b}
          \State \texttt{\textbf{DC}.Multicast}$(\<signReq>(b,\<batchId>,h))$~\label{seq-send-h}
          \State $\<wait>$ for $f+1$ responses
          $\<signResp>(\<batchId>,h,i,\sigma_i)$ with \(\sigma_i\)
          a valid signature\label{seq-wait}
          \State \<logger>.\<post>$(\<batchId>, h, \<BLSAggregate>(\bigcup_i \sigma_i))$ \label{seq-post}
         \State $\<batchId> \leftarrow \<batchId> +
         1$ \label{seq-increment-batchid}
         \State $\<pendingTxs> \leftarrow \epsilon$ \label{seq-when-clean-pendingTxs}
        \EndUpon \label{seq-when-end}
    \end{algorithmic}
  \end{algorithm}
\end{figure}


%% file: algorithms/dc-member.tex
\begin{figure}[h]
  \begin{algorithm}[H]
    \caption{\small DAC replica $i$.}
    \label{alg:dc-member}
    \small
    \begin{algorithmic}[1]
  	\State  \textbf{init:} $\<hashes> \leftarrow \emptyset$
        \Function{\textnormal{\<translate>}}{$\textit{id},h$}
           \If{$(\<id>,\_)\notin \<hashes>$}
             \<return> \<invalidId>
           \Else
             \If{$(\<id>,h)\notin \<hashes>$}
                \<return> \<invalidHash>
             \Else
               ~\<return> $\<hashes>(\<id>,h)$
             \EndIf
           \EndIf  
        \EndFunction
        \Receive{$\<signReq>( b,\<id>, h)$ from Sequencer} \label{dc-member-receive}
          \State \<assert> $h == \hash(b)$ and $(\<id>,\_) \notin \<hashes>$ \label{dc-member-sign-assert}
          \State $\<hashes> \leftarrow \<hashes> \cup \{((\<id>,h),b)\}$ \label{dc-member-sign-add}
          \State \textbf{send} $\<signResp>(\<id>,h,i, \<sign>_i(\<id>,h))$ to
          Sequencer \label{dc-member-sign-send}
        \EndReceive \label{dc-member-receive-end}
    \end{algorithmic}
  \end{algorithm}
 \end{figure}


%% file: SeqDecentralized.tex
\section{Arranger \#2: Full Decentralization}
\label{sec:seq-decentralized}

%
We present now a \emph{fully decentralized} arranger implementation
using Set Byzantine Consensus (see
Fig.~\ref{fig:arranger_decentralized}) that satisfies the correctness
properties listed in Section~\ref{sec:api-properties}.
In this implementation, operations \<add> and \<translate> are
completely decentralized.

\begin{figure}[!h]
    \centering
    \includegraphics[scale=0.40]{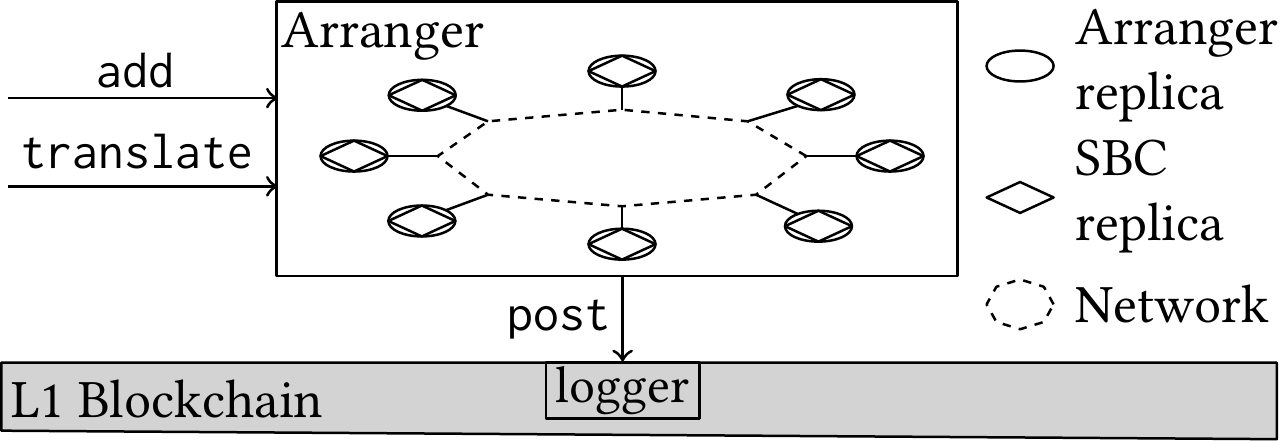}
    \caption{Fully decentralized arranger using SBC.}
    \label{fig:arranger_decentralized}
\end{figure}

Our implementation uses the sets decided by SBC rounds as batches.
Arranger replicas \emph{extend} SBC replicas instead of
modularly using an SBC as clients.
This is because Byzantine-tolerant blocks do not compose efficiently
as they need to implement defensive mechanisms to protect against
malicious clients and servers.
Clients interacting with an SBC replica cannot determine whether the
replica is honest or Byzantine and SBC replicas cannot determine
whether a client is dishonest.
As consequence, clients must contact multiple replicas to ensure that
the data they receive is the data agreed upon by all honest SBC
replicas.
This is not necessary when extending honest SBC replicas to build
honest arranger replicas, since replicas can trust the data maintain
internally.
In our fully decentralized arranger, all replicas perform both the roles of
sequencers and DAC members, and thus, honest arranger replicas are
honest SBC replicas,
honest sequencers, and honest DAC replicas.
Finally, arranger replicas take turns to post data to L1.

\subsection{Implementation}
Alg.~\ref{alg:arranger-sbc} shows the pseudo-code of an honest
arranger replica (here we only show code in addition to the SBC part
of the replica).
Replicas running Alg.~\ref{alg:arranger-sbc} keep a local map
\<hashes> storing inverse translations of hashes and a local set
$\<signatures>$ storing cryptographically signed batch tags.

\begin{figure}[h]
\input{algorithms/arranger-sbc}
\end{figure}

Distributed arranger replicas implement \<add> by adding transaction
requests as elements into their SBC replica to be included in a later
proposal and eventually in a decided set.
When a new set is decided by SBC, each replica locally receives a
\(\<SetDeliver>\) event with the round identifier and the consented set.
As mentioned above, honest arranger replicas extend honest SBC
replicas.
Therefore, honest replicas know that the set received in the event
\(\<SetDeliver>\) is the one that has been agreed upon by all other
honest SBC replicas.
Then, honest replicas order all new valid transaction requests in the
consented and packed these to form a new batch using function
\(\<tobatch>\).
Then, they update their local variable \<hashes>.
After that, honest replicas sign the new batch tag and broadcast
the signature.
Each replica receives a \(\<deliver>\) event informing about the
correctly cryptographically signed batch tags they received which are
stored in the local set $\<signatures>$.

Arranger replicas take turns posting legal batch tags into the
\<logger> contract in L1, for example slicing time and distribute the
time slices between the replicas, which is feasible because we assumed
the system is partially synchronous and slices can be reasonably long
intervals like seconds.
Event \<myturn> tells replicas when it is their turn to post a new
batch tag to L1.
This is a simple way to guarantee eventual progress.
Even if Byzantine replicas do not post legal batch tags in their
turns, the rounds will continue to advance until it is the turn of an
honest replica.
Also, note that Byzantine replicas cannot possibly collect enough
signatures of a fake batch tag to generate a certified batch tag so
either (1) they post a batch tag without enough signatures (which will
be discarded), or (2) they post the legal batch tag, or (3) they
remain silent.
Function $\<nextBatch>$ returns the next batch identifier to be posted
to the \<logger> contract which is the minimum identifier for which
there is no legal batch tag in the \<logger> L1 smart contract.
When it is the turn of arranger replica $\Rtt$ to post, and $\Rtt$ has
received enough signatures for the batch, $\Rtt$ computes a combined
signature $\sigma$ and posts the legal batch tag
$(\<batchId>,h,\sigma)$ into the \<logger> L1 smart contract.

\subsection{Proof of Correctness}%
\label{app:ProofsSeqDec}

Alg.~\ref{alg:arranger-sbc} implements a correct
arranger.
We inherit the assumptions from SBC that the number of Byzantine
replicas are less than one third of all arranger replicas, so for $N$
total replicas, the maximum number of Byzantine replicas is $f$ such
that $f < N/3$.

We first prove that a certified batch tag corresponds with the set
consented in an SBC round.
Formally, let \((\<batchId>, \hash(b),\sigma)\) be a certified batch
tag, then $\<tobatch>(S) = b$ where set \(S\) corresponds to the set
consented by SBC round number $\<batchId>$.

\begin{lemma}\label{l:batch-epoch}
  A batch of transaction requests in certified batch tag corresponds to
  transaction requests in a consented SBC set.
\end{lemma}

\begin{proof}
  A certified batch tags requires $f+1$ signatures, so at least one
  honest replica has signed the batch.
  Honest replicas only sign the hash of the set consented by the SBC
  service.
\end{proof}

From the previous lemma and properties \PrSBCValidity and
\PrSBCIntegrity, we obtain the arranger property \PrLegality.
Moreover, the arranger property \PrUniqueBatch follows from the
previous lemma and the property \PrSBCAgreement.

Byzantine replicas can re-post existing certified batch tags, but they
can not forge fake batches or fabricate certified batch tags because
this requires collecting $f+1$ signatures, which implies that
at least one honest replica signed.
Hence, all Byzantine replicas can do is to post the same batch tag
several times (perhaps with different collections of $f+1$
signatures).
Since the \<logger> smart contract only accepts the first legal batch
tag per identifier, no transaction is executed twice and repeated
batches are discarded.

\begin{lemma}\label{l:termination}
  All SBC consented sets are posted in the \<logger> L1 smart contract
  as legal batch tags.
\end{lemma}

\begin{proof}
  Let $\<batchId>$ be the minimum number for which there is no legal
  batch tag posted in the \<logger> contract, and let $S$ be the set
  consented in the $\<batchId>$-th SBC round of consensus.
  Property \PrSBCAgreement ensures that all honest replicas agree on
  the content of $S$.
  Honest arranger replicas cryptographically sign the hashes of the
  the sets agreed by SBC consensus, and broadcast their signature to
  all other arranger replicas.
  Therefore, all arranger replicas eventually receive the signature
  from honest replicas.
  Let \(\Rtt\) be the first honest replica that is in charge of
  posting batch tags to the logger after having received at least
  $f+1$ signatures of the hash of set $S$.
  There are two cases: either $\<nextBatch>$ returns $\<batchId>$ or
  it returns a larger identifier.
  In the former case, $\Rtt$ has enough signatures to generate a legal
  batch tag for $s$ and posts it in the \<logger> contract
  (lines~\ref{alg:arranger-sbc-upon-post-begin}-\ref{alg:arranger-sbc-upon-post-end}).
  In the latter case, a legal batch tag with identifier $\<batchId>$
  has already been posted corresponding to set \(S'\)~(see
  Lemma~\ref{l:batch-epoch}).
  By~\PrSBCAgreement, we have that $S'=S$.
\end{proof}

All transaction requests added to at least one honest replica
eventually appear in a legal batch tag, which follows from the
previous lemmas and property~\PrSBCCensorshipResistance.
Therefore, L2 users that contact some honest replica are guaranteed
\PrTermination.
L2 users can make sure that they contact at least one honest replica
by contacting $f+1$ replicas sequentially or in parallel.
Alternatively, clients can follow a lightweight protocol contacting
just one replica, then waiting for some time to observe whether the
transaction request is in a batch accepted by the \<logger> smart
contract.
If some time passes and the transaction is not seen, the client can
contact another replica.
This protocol guarantees that the transaction request will eventually
be included in a batch accepted by the \<logger> smart contract.
Moreover, since duplicated transaction requests will be discarded,
this protocol also guarantee that the transaction is executed exactly
once.

Finally, we prove that every hash in a legal batch tag posted by the
arranger can be resolved by at least one honest arranger replica,
satisfying Property \PrAvailability.

\begin{lemma}~\label{lem:one-correct} For every legal batch tag posted
  by the arranger, at least one honest replica $\Rtt$ signed it and
  therefore $\Rtt$ can resolve its hash.
\end{lemma}

\begin{proof}
  Since legal batch tags have $f+1$ signatures, at least one honest
  replica signed it.
  By Lemma~\ref{l:batch-epoch}, all legal batch tags correspond to SBC
  consented sets.
  When SBC consent on a set, it triggers an event \(\<SetDeliver>\)
  signaling the arranger replica to add the hash to its local map
  $\<hashes>$.
  Upon request, function $\<translate>$ uses $\<hashes>$ to obtain the
  batch that corresponds to a given hash.
  Honest replicas never delete information from $\<hashes>$, so for
  every legal batch tag posted at least one of the honest replicas
  that signed the batch tag can resolve its hash.
\end{proof}

Lemma~\ref{lem:one-correct} guarantees that, when posting legal batch
tags, at least one honest replica had signed the batch, and this
replica will be capable of performing the translation.
Eventually, \emph{all} honest replicas will know the contents of every
consented set so they will be able to translate as well.
Similar to the clients in the semi-decentralized arranger from
Section~\ref{sec:seqDC}, clients in this fully decentralized system
are guaranteed that after contacting all signing replicas of a
certified batch tag, they will receive the batch of transaction
requests.
Additionally, they can also employ different strategies when
contacting arranger replicas to translate hashes.

Since, we have proved that our arranger satisfies all correctness
properties of an arranger, we can conclude:

\begin{theorem}
  Arrangers consisting of \(n\) replicas, with \(f < \frac{n}{3}\)
  being Byzantine and the remaining \(n-f\) being honest replicas
  implementing Alg.~\ref{alg:arranger-sbc}, are correct
  arrangers.
\end{theorem}


%% file: algorithms/arranger-SBC.tex

  \begin{megaalgorithm}[H]
    \caption{\small Replica \(i\) implementation of \arranger using
      SBC}
    \label{alg:arranger-sbc}
    \small
    \begin{algorithmic}[1]
        \State  \textbf{init:} $\<hashes> \leftarrow \emptyset, \<signatures> \leftarrow \emptyset$
        \Function{\textnormal{\<add>}}{$t$} 
          \State $\<SBC>.\<Add>(t)$ 
        \EndFunction 
	\Function{\textnormal{\<translate>}}{$\textit{id},h$}
          \If{$(\<id>,\_)\notin \<hashes>$}
          \<return> \<invalidId>
          \Else
           \If{$(\<id>,h)\notin \<hashes>$} \<return> \<invalidHash>
          \Else ~\<return> $\<hashes>(\<id>,h)$
          \EndIf
          \EndIf
        \EndFunction 
        \Upon{$\<SBC>.\<SetDeliver>(\<id>,S)$}
        \State $b \leftarrow \<tobatch>(S)$
        \State $h \leftarrow \hash(b)$
        \State $\<hashes> \leftarrow \<hashes> \cup \{(\<id>,h) \mapsto b\}$
        \State $\<broadcast>(\<id>,h, \<sign>_i(\<id>,h))$
        \EndUpon
        \Upon{$\<deliver>(\<id>,h,\sigma_j)$ when \(\sigma_j\) is a
          valid signature}
        \State $\<signatures> \leftarrow \<signatures> \cup \{(id,h,\sigma_j)\}$
        \EndUpon
        \Upon{$\<myturn>$} \label{alg:arranger-sbc-upon-post-begin}
        \State $id \leftarrow \<nextBatch>()$
        \State \<wait> for $f+1$ signed batch tags $(\<id>,h,\sigma_j)$ in \<signatures>
        \State \<logger>.\<post>$(\<id>,h, \<BLSAggregate>(\bigcup_j \sigma_j))$ 
        \EndUpon  \label{alg:arranger-sbc-upon-post-end}
      \end{algorithmic}
    \end{megaalgorithm}


%% file: EmpiricalEvaluation.tex
\section{Empirical Evaluation}
\label{sec:empirical}

We provide now an empirical evaluation to asses whether the
decentralized arranger from Section~\ref{sec:seq-decentralized} can
scale to handle the current demand of L2s and beyond.

\subsection{Introduction}

There are two elements to assess this scalability: (1) the number
blocks and size of each block that Set Byzantine Consensus can
achieve, and (2) the new computations required to implement the
algorithm in Section~\ref{sec:seq-decentralized}.

\subsubsection*{Setchain and Blockchain Scalability}
Several implementations of SBC exist, including
RedBelly~\cite{Crain2021RedBelly},
Setchain~\cite{capretto24improving}, and
ZLB~\cite{ranchalPedrosa2024ZLB}.
All these implementations report a throughput exceeding 12,000
transactions per second (TPS).
In comparison, the cumulative throughput of all Ethereum L2s amounts
to \LdosTPS{} TPS, \footnote{Data obtained on Feb. 2025 from
  l2beat \url{https://l2beat.com/scaling/activity}} which can be taken
as the baseline of L2s \emph{current demand}.\footnote{The available
  documentation of L2s does not present theoretical or empirical
  limitations.}

\subsubsection*{Evaluating Additional Building Blocks Locally}
Our approach extends SBC, and thus, we focus on assessing the
efficiency of each new component forming the decentralized arranger:
compressing, hashing, signing batches, verification and aggregation of
signatures, and translation of hashes into compressed batches.
We aim to evaluate empirically that these components impose a
negligible computational overhead on top of the intrinsic
set-consensus algorithm.
All these components run locally and are easily parallelizable.

\subsection{Empirical Evaluation}

\subsubsection*{Setup}
All experiments were carried out on a Linux machine with $256$~GB of
RAM and $72$~virtual 3GHz-cores (Xeon Gold~6154) running
Ubuntu~20.04.
We used a sequence of 40,000 real transaction requests from Arbitrum
One as our data set.

For the evaluation, we used Brotli~\cite{Alakuijala18brotli} as
compression algorithm and the root of Merkle trees~\cite{Merkle88} of
batches as hash function \(\<hash>\).
For signing, verifying and aggregating signatures, we employed the BLS
signature scheme implemented in~\cite{ArbitrumNitroGithub}.

\subsubsection*{Hypothesis and Experiments}
%
\newcommand{\Hypo}[1]{\textup{\textbf{#1}}\xspace}
\newcommand{\HSize}{\Hypo{H.Size}}
\newcommand{\HHash}{\Hypo{H.Hash}}
\newcommand{\HCompress}{\Hypo{H.Compress}}
\newcommand{\HSign}{\Hypo{H.Sign}}
\newcommand{\HAgg}{\Hypo{H.Agg}}
\newcommand{\HVer}{\Hypo{H.Ver}}
\newcommand{\HTrans}{\Hypo{H.Trans}}
%
We intend to answer empirically the following research question:
\emph{these building blocks impose a negligible computational overhead
  on top of the set-consensus algorithm.}

In other words, our hypothesis is that the performance offered by the
implementations of SBC is not affected by:
\begin{itemize}
\item(\HHash) hashing batches of transactions,
\item (\HCompress) compressing batches of transactions,
\item (\HSign) signing hashes of batches,
\item (\HAgg) aggregating signatures,
\item (\HVer) verifying signatures,
\item (\HTrans) translating hashes into batches.
\end{itemize}
Additionally, signed batch tags size is significantly smaller than
their corresponding compressed batch (hypothesis \HSize).
We report on each hypothesis separately.

\begin{itemize}

\item{\HSize.}
%
  To evaluate \HSize, we measured the size of both compressed batches
  and batch tags ranging from $400$ to 4,400 transactions per
  batch.\footnote{When we run the experiments Arbitrum Nitro batches
    typically contained around $800$ transaction requests, while in
    Arbitrum AnyTrust, batches contained approximately 4,000
    transaction requests. As of February 2025, Arbitrum Nitro batches
    typically contain around 3,000 transaction requests, while in
    Arbitrum AnyTrust, batches contain approximately $100$ transaction
    requests.}
For each batch size, we shuffled transaction requests from our data
set, split them into 10 batches of equal size, and assigned them
unique identifiers.
%
For the compression of each batch, we used the Brotli compression
algorithm~\cite{Alakuijala18brotli}, currently used in Arbitrum
Nitro~\cite{ArbitrumNitro}.
Additionally, for each batch, we computed its Merkle
Tree~\cite{Merkle88} with their transactions as leaves.
To create batch tags, we signed each batch Merkle Tree root along with
an identifier using a BLS Signature scheme~\cite{Boneh2001Short}
implemented in~\cite{ArbitrumNitroGithub}.
Our experiments show there is a ratio of \emph{more than two orders}
of magnitude between compressed batches and batch tags with the same
number of transactions.
Fig.~\ref{fig:size} shows the average size of compressed batches of
400 transactions is approximately 80,000~Bytes, whereas the
average size for 4,400 transactions is around 780,000~Bytes.
In contrast, the average size of signed batch tags remained around
$480$~Bytes independently of the number of transactions within the
batch.
Therefore, hypothesis \HSize holds.

\begin{figure}[h]
  \includegraphics[scale=0.23]{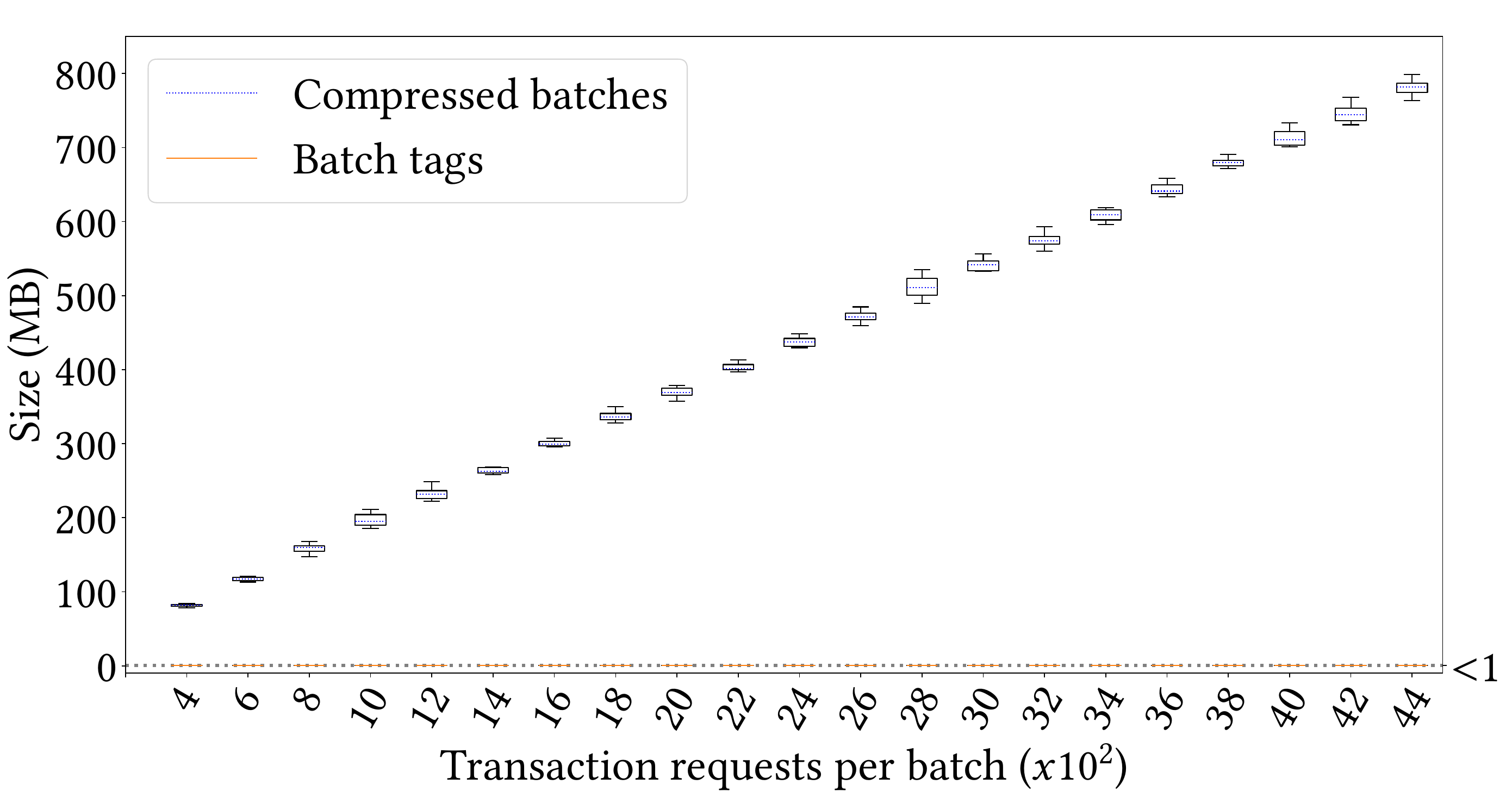}
\caption{Average size of compressed batches and batch tags for varying
  numbers of transaction requests per batch.}
 \label{fig:size}
\end{figure}
\end{itemize}

For the rest of the experiments, we loaded all the required
information into local memory and applied each procedure for a
duration of $1$~second, cycling over the input as much as necessary.
We ran each experiment $10$ times and report the average performance.
Fig.~\ref{fig:tps-hash-compress-translate}, Fig.~\ref{fig:aggregation},
and Fig.~\ref{fig:ver-parallel} summarize the results.

\begin{itemize}
\item{\HHash{} and \HCompress{}.}
To test compressing and hashing procedures, we use the full data set
of 40,000~transaction requests.
We again ranged from $400$ to 4,400 transaction requests per batch.
Fig.~\ref{fig:tps-hash-compress-translate} shows that hashing can
scale from 260,000 to 275,000~TPS, depending on the batch size,
while the compression procedure achieves a throughput ranging from
75,000~TPS for batches of $400$ transaction requests to 50,000~TPS
for batches of 4,400~transaction requests.
Both hashing and compression procedures significantly outperform the
SBC implementation throughput of 12,000~TPS, and thus, these
procedures do not impact the overall throughput.
Therefore, hypotheses \HHash and \HCompress hold.

\begin{figure}[h]
   \includegraphics[scale=0.34]{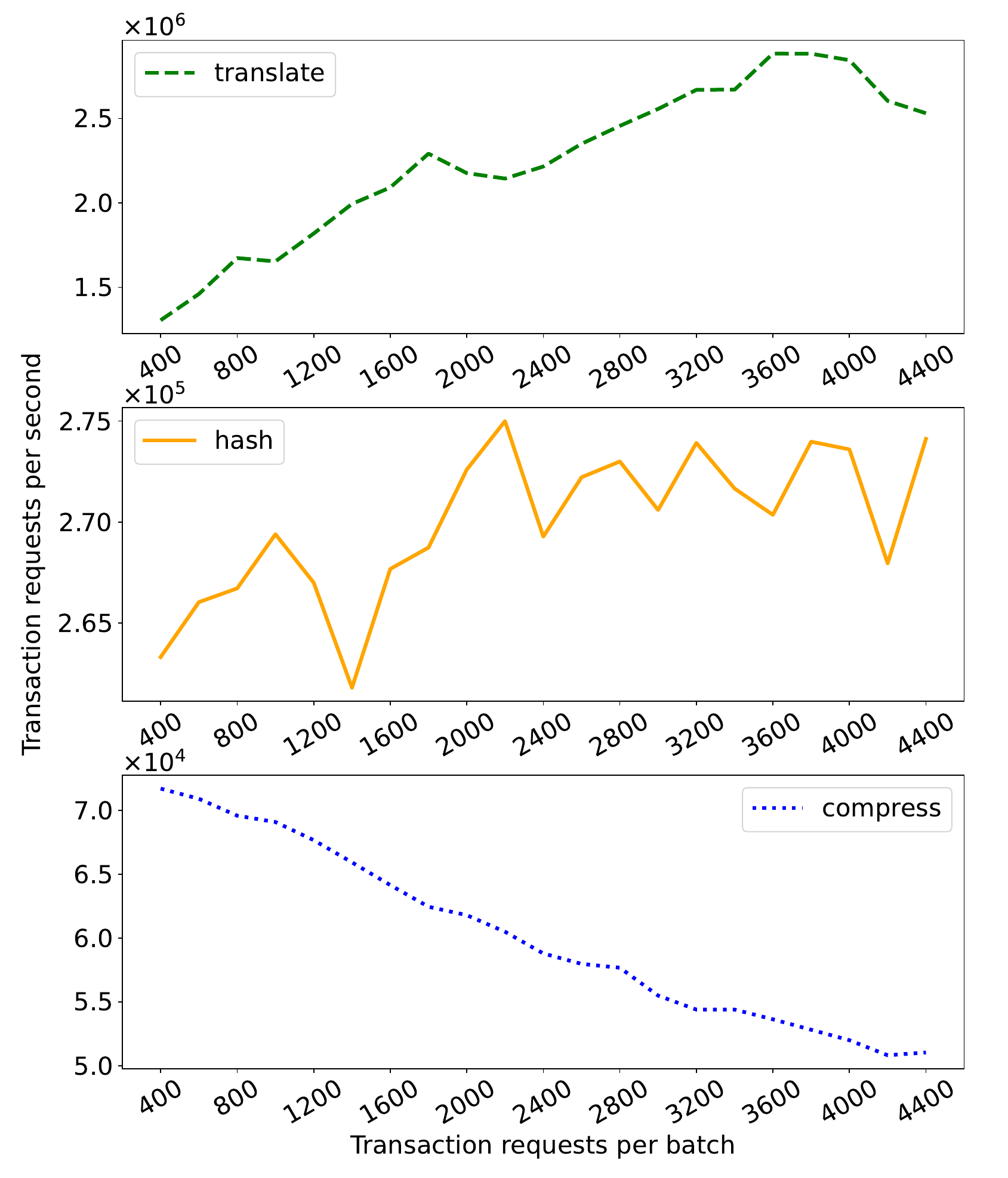}
   \caption{Throughput of compressing, hashing and translating
     procedures for varying number of transaction requests per batch.}
 \label{fig:tps-hash-compress-translate}
\end{figure}

\item{\HSign.}
To evaluate the signing procedure performance, we used a file
containing $50$ pairs batch hash-identifier.
We assess the performance of signing each pair using the provided
secret key with the BLS signature scheme implemented
in~\cite{ArbitrumNitroGithub}.
Our results show that our system can sign approximately 2,300
hash-identifier pairs per second.
Considering a throughput of 12,000~TPS and batches of at least $400$
transaction requests, no more than $30$ batches per second will be
generated.
Thus, each arranger replica can sign more than $75$ times the number
of batches required, empirically supporting hypothesis \HSign.

\item{\HAgg{}.}
The performance of the signature aggregation procedure depends on the
number of signatures aggregated, which is bounded by the number $N$ of
arranger replicas.
We studied the performance of signature aggregation ranging from \(8\)
to \(256\) replicas\footnote{As reference, currently, the DAC of
  Arbitrum AnyTrust comprises $7$
  members
  \url{https://docs.arbitrum.foundation/state-of-progressive-decentralization\#data-availability-committee-members}.}.
Our data set consists of a directory with $N$ files.
Each file contains $50$ signed hash-identifier pairs signed by the
same signer.
The list of pairs hash-identifier is the same across all files.
We measured the performance of aggregating $N$ signatures of the same
hash-identifier pair using BLS~\cite{ArbitrumNitroGithub}.
Fig.~\ref{fig:aggregation} shows an aggregation ratio ranging from
5,000~procedures per second, when aggregating $256$~signatures in
each invocation, to 115,000~per second when aggregating
$8$~signatures.
Assuming at most $30$ batches per second are generated, the
aggregation procedure can easily handle its workload, even for
$N = 256$ arranger replicas, empirically validating hypothesis \HAgg.

\begin{figure}[h]
  \includegraphics[scale=0.47]{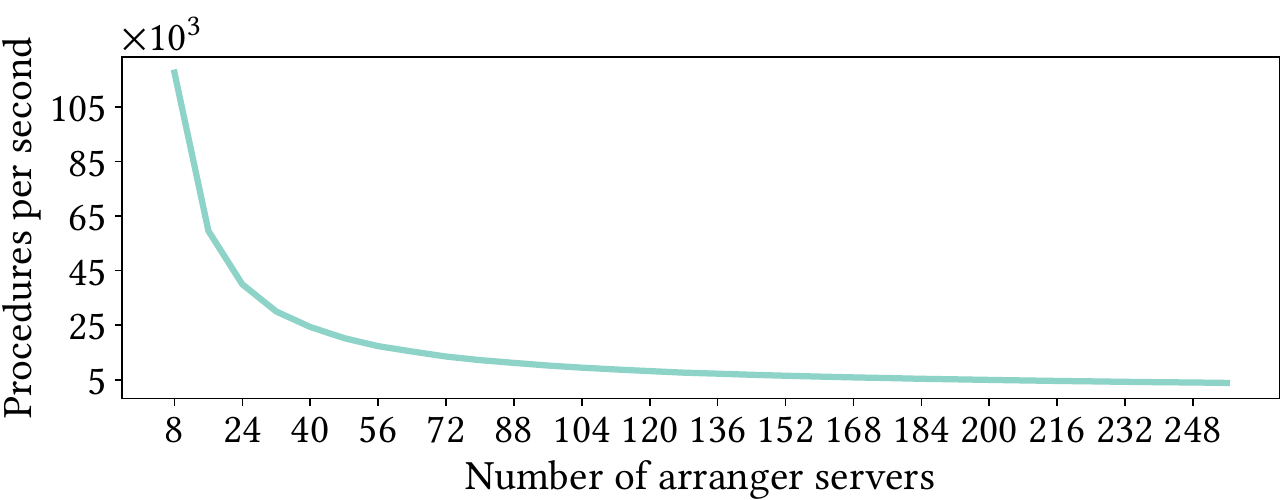}
  \caption{Signature aggregations per second, in terms of the number
    of signature aggregated per procedure.}
 \label{fig:aggregation}
\end{figure}

\item{\HVer{}.}
To assess the performance of the signature verification procedure, we
used a file containing $50$~batch tags and the public key of the
signer of all hashes.
We analyzed the performance of verifying that each signature was
generated by signing the corresponding hash-identifier pair using the
corresponding private keys.
Fig.~\ref{fig:ver-parallel} shows that signature verification can be
performed at a rate of approximately $600$~signatures per second.
This rate is sufficient to verify all signatures generated by $N=128$
replicas, considering a throughput of 12,000~TPS and batches
containing 4,400 transaction requests (dotted line in
Fig.~\ref{fig:ver-parallel}).
Replicas must verify the highest number of signatures when the
arranger is maintained by $256$ replicas and batches contain only
$400$ transaction requests, which amounts to verifying 7,680
signatures per second, surpassing the throughput of the signature
verification procedure (dashed line in Fig.~\ref{fig:ver-parallel}).
However, verifying signatures can be done in parallel because each
signature can verified independently from the other
signatures.\footnote{All other local procedures described can use
  parallelism easily as well.}

Our empirical evaluation suggests that using $16$~parallel processes
to verify signatures achieves a rate exceeding 8,500~signature
verified per second, sufficiently for $N=256$ replicas and batches
containing $400$ transaction replicas at 12,000~TPS.
Consequently, the verifying procedure does not influence the
throughput, and hypothesis \HVer holds.

\begin{figure}[h]
  \includegraphics[scale=0.47]{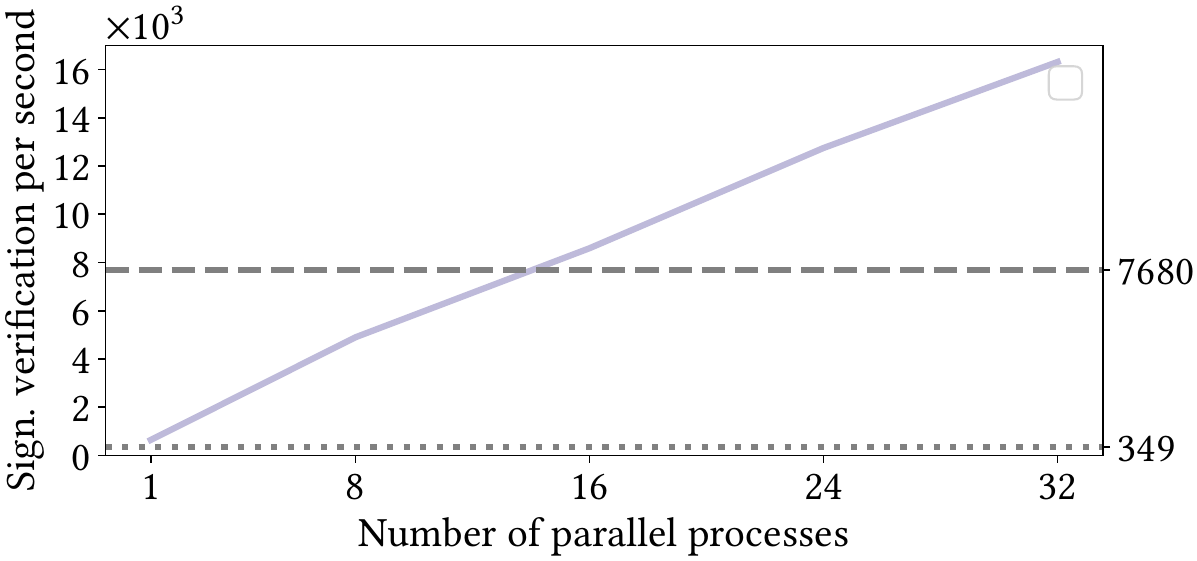}
  \caption{Signature verification per second.
    The dotted line
    represents the maximal number of signatures that each arranger
    replica must verify (for 12,000~TPS, the batch size of
    4,400 and $128$ replicas). The dashed line represents the maximal
    number of signatures that each arranger replica must verify (for
    12,000~TPS, batch size of $400$ and $256$ replicas).
  }
\label{fig:ver-parallel}
\end{figure}

\item{\HTrans{}.}
Finally, we implemented a \emph{translation server} that takes two
input files and generates a dictionary mapping hashes to their
compressed data.
One file contains hashed batches together with identifiers while the
other contains compressed batches with identifiers.
The translation server opens a local TCP connection port to receive
hash translation requests receiving an identifier and returning a
compressed batch according to the dictionary.

Our experiments involve making sequential requests to a single
translation server asking to translate input hashes for a duration of
$1$ second.
All batches contained the same number of transaction requests, ranging
from $400$ to 4,400 transactions.

Fig.~\ref{fig:tps-hash-compress-translate} shows that the translation
for batches of $400$ transactions can be performed at a rate of about
1,300,000~TPS.
The translation throughput increases as the number of transaction
requests per batch increases, peaking at nearly 2,900,000~TPS for
batches of 3,800~transactions.
However, the throughput decreases slightly after that point, due to
the increased size of compressed batches.
The translation procedure can easily handle two orders of magnitude
more transaction requests than the throughput of implementation of
SBC, confirming hypothesis~\textbf{\HTrans{}}.
\end{itemize}

\subsection{Experiment Conclusions}
The conclusions of our experiments are the following:
\begin{enumerate}
  \item hashing and
signing batches of transaction requests take significantly less space
than compressing batches of transaction requests, and
\item local
procedures are very efficient and do not limit the throughput provided
SBC implementations, which is the main bottleneck for scalability.
\end{enumerate}

Since the overhead of the procedures added to the set-consensus
algorithm is negligible, our decentralized arranger does not introduce
any significant performance bottlenecks to SBC implementations, even
if the workload of L2 systems was increased by two orders of magnitude
and there were 256 arranger replicas.

Finally the reduction in space, which remains constant as batch size
increases, contributes to reduce the L1 gas needed.



%% file: RelatedWork.tex
\section{Related Work}
\label{sec:related-work}

While the limitations of centralized arrangers in L2s are known, the
development of decentralized arrangers is still in its early
stages~\cite{motepalli2023sok}.

Many L2s implement an arranger as centralized sequencer and separate
data availability committee~\cite{l2beat}, as the one we presented in
Section~\ref{sec:seqDC}.
However, to the best of our knowledge, Metis~\cite{metisDecentralized}
is the only L2 that currently implements a decentralized arranger.
Metis is based on proof-of-stake and uses
Tendermint~\cite{buchman2016tendermint} to select a rotating leader.
The leader collects transaction requests and generates batches.
These batches are signed using multi-party computation and, when
enough signatures are gathered, the leader posts a signed batch in L1.
However, Metis does not provide a proof of correctness.
This is in contrast to our fully decentralized arrange, presented in
Section~\ref{sec:seq-decentralized}, which we proved correct under the
assumption that less than one-third of the replicas are Byzantine.
%

Other attempts to implement arrangers include Radius~\cite{Radius},
Espresso~\cite{espressoSequencer} and Astria~\cite{Astria}.
Radius also uses leader election, based on the RAFT
algorithm~\cite{ongaro2014raft}, which is not Byzantine-tolerant.
Radius replicas remain consistent and can reach consensus even when the
leader fails, but~\cite{Radius} does not discuss what happens when
replicas refuse to disclose block contents to users.
Radius allows encryption of transactions to prevent maximum
extractable value (MEV), but we do not address this problem here.
Both Espresso and Astria use one protocol for sequencing transaction
requests and a different protocol for data availability.
In particular, Espresso uses HotStuff2~\cite{hotstuff} as consensus
mechanism and its own Data Availability layer known as
\emph{tiramisu}, while Astria uses CometBFT~\cite{cometBFT} for
sequencing transaction requests and leverages Celestia Data
Availability Service~\cite{celestia}.
We tackled both challenges in a single protocol because
Byzantine-resilient solutions do not compose efficiently (see
e.g.~\cite{capretto24improving}).

Finally, Malachite~\cite{malachite} is a flexible Byzantine
Fault-Tolerant engine currently under development aimed at
decentralizing only the sequencer component of arrangers.
Malachite could be directly applied to both Optimistic Rollups and ZK
Rollups, (where batches are posted compressed and a DAC is not
required to translate hashes) or combined witha DAC service to obtain
an arranger.


%% file: Conclusion.tex
\section{Conclusion}
\label{sec:conclusion}

Layer 2 Blockchains aim to improve the scalability of current smart
contract based blockchains, offering a much higher throughput without
modifying the programming logic and user interaction.
Two crucial elements of all L2s schemes are the \emph{sequencer},
which receives and orders transaction requests from users, packs them
into batches and sends the hash of the batch to L1, and the \emph{data
  availaibility committe} which make sure that the data corresponding
to the hash is available.
We introduced the notion of \emph{arranger} which combines the
sequencer and DAC in a single service.
Current implementations of arrangers in most L2s are based on
centralized sequencers, either posting compressed batches~(ZK-Rollups
and Optimistic Rollups), or hashes of batches with a fixed collection
of servers providing data translation~(Validiums and Optimiums).
The resulting L2s are fully controlled by the centralized sequencer,
that has full power for censoring transactions and the ability to post
fake batches without penalties.

In this paper, we rigorously defined the correctness criteria for
arrangers and presented a fully decentralized arranger
proof-of-concept based on the efficient Set Byzantine Consensus
protocol.
We showed our solution is correct when the portion of Byzantine
replicas is less than one third.
We implemented all building blocks required to extend SBC into an
honest arranger replicas, and presented an empirical evaluation that
shows that the overhead will not reduce the throughput of SBC
implementations to handle several times the current demand.

\subsection*{Future work}

Future work includes extending a production SBC implementation
transforming it into a decentralized arranger implementation, allowing
all components to run together and enabling a stronger empirical
evaluation and study other factors, such as latency.

Throughout the paper, we assumed that the number of Byzantine arranger
replicas is bounded by \(f\), a known fraction of the total number of
arranger replicas.
However, we did not explore what the incentive for an arranger
replica to be honest or Byzantine is.
This requires both (1) a system of rewards for replicas that make the
L2 blockchain progress (for example, extending the assignment of
rewards to STFs in optimistic rollups) and (2) a collection of new
fraud-proof games to punish replicas that are proven to behave
dishonestly.
This requires that replicas place a stake to account for the
penalties, which is recovered after claims consolidate.
An interesting avenue for future work is to study incentives and
punishments to align the behavior of rational arranger with the honest
arranger replicas, creating a system that is \emph{incentive
  compatible}~\cite{Ledyard1989}.
